\documentclass[conference]{IEEEtran}

\usepackage{graphicx}
\usepackage{psfrag}
\usepackage{amsmath}

\usepackage{epsfig,color}
\usepackage{amssymb}
\usepackage{cite}
\usepackage{dsfont}
\newtheorem{proposition}{Proposition}
\newtheorem{lemma}{Lemma}
\newtheorem{example}{Example}
\makeatletter

\begin{document}

\title{Buffers Improve the Performance of Relay Selection\\[-3mm]}

\author{\IEEEauthorblockN{Aissa Ikhlef, Diomidis S. Michalopoulos, and Robert Schober}
\IEEEauthorblockA{The University of British Columbia, Vancouver, Canada\\
E-mail: \{aikhlef,\,dio,\,rschober\}@ece.ubc.ca}\\[-8mm]}

\maketitle

\begin{abstract}
We show that the performance of relay selection can be improved by employing relays with buffers. Under the idealized assumption that no buffer is full or empty, the best source-relay and the best relay-destination channels can be simultaneously exploited
by selecting the corresponding relays for reception and transmission, respectively. The resulting relay selection scheme is referred to as max-max relay selection (MMRS). Since for finite buffer sizes, empty and full buffers are practically unavoidable if MMRS is employed, we propose a hybrid relay selection (HRS) scheme, which is a combination of conventional best relay selection (BRS) and MMRS. We analyze the outage probabilities of MMRS and HRS and show that both schemes achieve the same diversity gain as conventional
BRS and a superior coding gain. Furthermore, our results show that for moderate buffer sizes (e.g.~30 packets) HRS closely approaches the performance of idealized MMRS and the performance gain compared to BRS approaches 3 dB as the number of relays
increases.
\end{abstract}

\IEEEpeerreviewmaketitle

\section{Introduction}
In cooperative networks with multiple relays \cite{Laneman_TIT04}, where a number of relays assist a source in transmitting information to a destination, relay selection techniques have gained a lot of interest \cite{Bletsas_KRL_JSAC06}. Relay selection is attractive because of its high performance, efficient use of power and bandwidth resources, and simplicity.  For example, simple relay selection schemes  \cite{Bletsas_KRL_JSAC06} can achieve the same diversity order as more complex cooperative schemes employing space time block coding \cite{Laneman_TIT03} or orthogonal channels \cite{Laneman_TIT04}. Many different schemes for single relay selection have been proposed in the literature, see e.g.~\cite{Bletsas_KRL_JSAC06, Bletsas_SW_COMLet07, Zhao_AL_TWC07, Dio_KTM_TWC08} and references therein. All these schemes have in common that the selected relay receives a packet from the source and retransmits it to the destination. The two most common schemes are the bottleneck and maximum harmonic mean based \textit{best relay selection} (BRS) schemes \cite{Bletsas_KRL_JSAC06} and their performances have been extensively studied in the literature \cite{Bletsas_KRL_JSAC06,Dio_K_TWC08,
Ikki_A_TCOM10, Ikki_A_EURASIP_JASP08}. Here, we adopt bottleneck based BRS \cite{Bletsas_KRL_JSAC06} as a benchmark for the proposed schemes.

In order to overcome the limitations imposed by using the same relay for reception and transmission, we propose to employ relays with buffers for applications that are not delay sensitive. If the relays have buffers, the relays with the best source-relay channel and the best relay-destination channel can be selected for reception and transmission, respectively. The corresponding selection scheme is referred to as \textit{max-max relay selection} (MMRS). For MMRS, we make the idealistic assumption that the buffer of the relay selected for reception (transmission) is not full (empty), which is only possible for buffers of infinite size. For the practical case of finite buffer sizes the buffer of a relay may become empty (full) if the channel conditions are such that the relay is selected repeatedly for transmission (reception) but not for reception (transmission).

To overcome this limitation, we propose a hybrid relay selection (HRS) scheme, which is a combination of conventional BRS and MMRS. In particular, for HRS, if the buffer of the relay selected for reception (transmission) is full (empty), BRS is employed; otherwise,
MMRS is used. Although both MMRS and HRS can be combined with both amplify-and-forward and decode-and-forward (DF) relaying, in this paper, we only consider DF relays and derive the corresponding outage probabilities. Analytical and simulation results
establish the superiority of MMRS and HRS compared to BRS.

We note that relays with buffers have been considered before in \cite{Xia_Fan_Thompson_Poor_TWC08} and \cite{Zlatanov_Schober_Globecom11} to improve the throughput of simple three-node networks consisting of a source, a destination, and a single relay. However, to the best of our knowledge, relays with buffers have not been considered in the context of relay selection before.

The remainder of this paper is organized as follows. In Section \ref{system_model}, the system model is presented, and MMRS and HRS are introduced. An outage analysis of MMRS and HRS is provided in Section \ref{Outage_analysis}. In Section \ref{results}, numerical results are presented, and conclusions are drawn in Section \ref{conclusion}.
\section{System model and relay selection}
\label{system_model}
In this section, we present the system model, briefly review BRS, and introduce the proposed MMRS and HRS schemes.
\subsection{System Model}
We consider a relay network with one source node, $S$, one destination node, $D$, and $N$ half-duplex DF relays, $R_1, \ldots,R_N$. Each relay is equipped with a buffer and transmission is organized in two time slots. In the first time slot, the
relay selected for reception receives a packet from the source node and stores it in its buffer. In the second time slot, the relay selected for transmission forwards a packet from its buffer to the destination node.

We assume that a direct link between the source and the destination does not exist or, if it does exist, it is not exploited for simplicity of implementation. Let $g_{i}$ and $h_{i}$, $i=1,\ldots,N$, denote the $S$-$R_i$ and $R_i$-$D$ channel gains, respectively.
We assume that the channel coefficients $g_{i}$ and $h_{i}$ are mutually independent zero-mean complex Gaussian random variables (Rayleigh fading) with variances $\sigma_{g_i}^2$ and $\sigma_{h_i}^2$, respectively. Moreover, we assume that the
transmission is organized in packets and the channels are constant for the duration of one packet and vary independently from one packet to the next (block fading model). This behavior can be achieved through frequency hopping between packets.
Let $\gamma_{g_i}\triangleq |g_i|^2\frac{E_s}{N_0}$ denote the instantaneous signal-to-noise ratio (SNR) between the source and relay $R_i$ and $\gamma_{h_i}\triangleq |h_i|^2\frac{E_s}{N_0}$ the instantaneous SNR between $R_i$ and the destination.
Here, $E_s$ is the energy available at the transmitting nodes and $N_0$ is the variance of the zero mean additive white Gaussian noise (AWGN) at the receiving nodes. $\gamma_{g_i}$ and $\gamma_{h_i}$ are exponentially distributed with parameters
$1/\overline{\gamma}_{g_i}$ and $1/\overline{\gamma}_{h_i}$, respectively, where $\overline{\gamma}_{g_i}\triangleq E[\gamma_{g_i}]=\sigma_{g_i}^2\frac{E_s}{N_0}$, $\overline{\gamma}_{h_i}\triangleq E[\gamma_{h_i}]=\sigma_{h_i}^2\frac{E_s}{N_0}$,
and $E[\cdot]$ denotes expectation.

We assume that the destination node has perfect channel state information (CSI) and selects the relays for transmission and reception. The destination node feeds back the information about the selected relays to all relays via an error-free feedback channel.
%
\subsection{Best Relay Selection}
\label{BRS_scheme}
The BRS scheme achieves full diversity by selecting one relay out of the $N$ available relays. This relay is then used for reception and transmission. The selected best relay, $R_b$, has the best bottleneck link \cite{Bletsas_KRL_JSAC06}, i.e,
\begin{equation}\label{eq_1.1}
    b\triangleq \mathrm{argmax}_{i=1,...,N}~~\min\{\gamma_{g_i},\gamma_{h_i}\}.
\end{equation}
\subsection{Max-Max Relay Selection}
\label{MMRS_scheme}
The relay selected according to the criterion in (\ref{eq_1.1}) may not simultaneously enjoy the best source-relay and the best relay-destination channels. If the relays are equipped with buffers, they can store the packets received from the source and do not have to re--transmit
them immediately in the next time slot. As a result, it is possible to use the relay with the best source-relay channel for reception and the relay with the best relay-destination channel for transmission. Thus, in the resulting MMRS scheme, the best relay for reception,
$R_{br}$, is selected based on
\begin{equation}\label{eq_1}
    br\triangleq \mathrm{argmax}_{i=1,...,N}~~\gamma_{g_i}
\end{equation}
and the best relay for transmission, $R_{bt}$, is selected according to
\begin{equation}\label{eq_2}
     bt \triangleq \mathrm{argmax}_{i=1,...,N}~~\gamma_{h_i}.
\end{equation}
For MMRS to work properly, the buffer of no relay can be empty or full at any time such that all relays have always the option of receiving and transmitting. Clearly, for buffers of finite size this may not be possible since a buffer may become empty (full) if a relay enjoys repeatedly
the best relay-destination (source-relay) link but never the best source-relay (relay-destination) link. To overcome this problem, in the next subsection, we combine MMRS with BRS.
%
\subsection{Hybrid Relay Selection}
\label{HRS_scheme}
If buffer over- and underflows are to be avoided, the relay selection criterion cannot only depend on the channel status as in MMRS but also has to take into account the status of the buffer. The basic idea is to use BRS if either the buffer of the relay selected for reception is full or the buffer of the relay selected for transmission is empty. In all other cases, MMRS is used. We assume that all buffers have $L_b$ elements and each element can store one packet. We denote the number of elements of relay $R_i$'s buffer that are full by $N_{e,i}$. For HRS, the best relay for reception, $R_{\overline{br}}$, is selected according to
\begin{equation}\label{eq_2.1}
    \overline{br}=\left\{
        \begin{array}{ll}
            b,& \mathrm{if}~N_{e,br}=L_b-1~\mathrm{or}~N_{e,bt}=0, \\
            br,&\mathrm{otherwise}, \\
        \end{array}
        \right.
\end{equation}
and the best relay for transmission, $R_{\overline{br}}$, is selected according to
\begin{equation}\label{eq_2.1.1}
    \overline{bt}=\left\{
        \begin{array}{ll}
            b,& \mathrm{if}~N_{e,br}=L_b-1~\mathrm{or}~N_{e,bt}=0, \\
            bt,&\mathrm{otherwise}, \\
        \end{array}
        \right.
\end{equation}
where $b$, $br$, and $bt$ are defined in (\ref{eq_1.1}), (\ref{eq_1}), and (\ref{eq_2}), respectively. In (\ref{eq_2.1}) and (\ref{eq_2.1.1}), we always leave one element of each buffer empty so that each relay is always able to receive in case it is selected for reception
in the BRS mode in the next transmission interval.

We note that for both MMRS and HRS, since different packets may be stored at different relays for different amounts of time, the packets transmitted by the source may arrive at the destination node in an order different from the order at the source node.
The original order can be restored at the destination node if the order information is contained in the preamble of the packet. Furthermore, MMRS and HRS introduce a delay in the network. This issue will be investigated in Section \ref{results}.
%
\section{Outage Probability Analysis }
\label{Outage_analysis}
In this section, we study the outage probability of MMRS and HRS with DF relays. The outage probability is defined as the probability that the output SNR, $\gamma_{b}$, falls below a certain SNR threshold, $\gamma\triangleq  2^{2R}-1$,
above which error-free transmission with rate $R$ is possible \cite{Simon_Alouini_B00}, i.e.,
\begin{equation}\label{eq_2.2}
   P_{out}\triangleq P(\gamma_{b}\leq\gamma),
\end{equation}
where $P(A)$ denotes the probability of event $A$.
Before we consider MMRS and HRS, we first briefly review the outage probability of BRS, which will be useful for computation of the outage probability of HRS.
\subsection{Best Relay Selection}
\label{Outage_BRS}
For BRS and DF relays, $\gamma_{b}\triangleq \max_i\{\min(\gamma_{g_i},\gamma_{h_i})\}$. Thus, based on (\ref{eq_2.2}), we obtain \cite{Dio_K_TWC08}
\begin{equation}\label{eq_2.3}
   P_{out}^{BRS}=\prod_{i=1}^{N}\left(1-\exp\left(-\frac{\gamma}{\bar{y}_i}\right)\right),
\end{equation}
where $\bar{y}_i\triangleq \left(1/\bar{\gamma}_{g_i}+1/\bar{\gamma}_{h_i}\right)^{-1}$.

At high SNR and assuming independent and identically distributed (i.i.d.) fading for both links, i.e.,  $\overline{\gamma}_{g_i}=\overline{\gamma}_{h_i}=\overline{\gamma}$ and $\bar{y}_i=\frac{\bar{\gamma}}{2}$, the outage probability can be simplified to
\begin{equation}\label{eq_2.4}
   P_{out}^{BRS}\approx\left(\frac{2\gamma}{\bar{\gamma}}\right)^N.
\end{equation}
Expressing the outage probability now in terms of the diversity gain $G_d$ and coding gain $G_c$, i.e., $ P_{out}\approx\left(G_c \bar{\gamma}/\gamma\right)^{-G_d}$ \cite{Tarokh_Seshadri_Calderbank_IT98}, we observe that the diversity gain of BRS is
$G_d^{BRS}=N$ and its coding gain is $G_c^{BRS}=1/2$.
\subsection{Max-Max Relay Selection}
\label{Outage_MMRS}
For MMRS and DF relays, we have $\gamma_{b}\triangleq\min\{\gamma_{g_b},\gamma_{h_b}\}$, where $\gamma_{g_b}\triangleq \max_{i=1\ldots N}\gamma_{g_i}$ and $\gamma_{h_b}\triangleq \max_{i=1\ldots N}\gamma_{h_i}$.
Hence, the outage probability is given by
\begin{align}\label{eq_3}
    P_{out}^{MMRS}&= P\big(\min \{\gamma_{g_b},\gamma_{h_b}\}\leq \gamma\big)\nonumber\\
    &\hspace*{-3mm}= 1-P\big(\gamma_{g_b}> \gamma\big)P\big(\gamma_{h_b}> \gamma\big)\nonumber\\
    &\hspace*{-3mm}= 1-\Big[1-P\big(\gamma_{g_b}\leq \gamma\big)\Big] \Big[1-P\big(\gamma_{h_b}\leq \gamma\big)\Big]\\
    &\hspace*{-3mm}= 1-\Big[1-\prod_{i=1}^{N}\Big(1-e^{-\frac{\gamma}{\overline{\gamma}_{g_i}}}\Big)\Big]
    \Big[1-\prod_{i=1}^{N}\Big(1-e^{-\frac{\gamma}{\overline{\gamma}_{h_i}}}\Big)\Big].\nonumber
\end{align}
In case of i.i.d.~fading for both links, i.e., $\overline{\gamma}=\overline{\gamma}_{g_i}=\overline{\gamma}_{h_i}$, $i=1,...,N$, (\ref{eq_3}) simplifies to
\begin{eqnarray}\label{eq_4}
    P_{out}^{MMRS}= 1-\Big[1-\Big(1-e^{-\frac{\gamma}{\overline{\gamma}}}\Big)^N\Big]^2.
\end{eqnarray}
If we assume furthermore that the SNR is high and use the approximation $1-e^{-x}\approx x$, $x\to 0$, we obtain
\begin{eqnarray}\label{eq_5}
    P_{out}^{MMRS}&\approx& \Big(2^{\frac{1}{N}}\frac{\gamma}{\overline{\gamma}}\Big)^{N}.
\end{eqnarray}
From (\ref{eq_5}), we observe that MMRS achieves a diversity gain of $G_d^{MMRS}=N$ and a coding gain of $G_c^{MMRS}=2^{-\frac{1}{N}}$. Thus, in contrast to BRS, the coding gain of MMRS increases with the number of relays.

Interestingly, BRS and MMRS have the same diversity gain. However, MMRS achieves a higher coding gain for $N\ge 2$ relays. For a large number of relays, MMRS yields an SNR gain of $\lim_{N\to\infty}10\log_{10}(G_c^{MMRS}/G_c^{BRS})=3$ dB compared to BRS.
%
\subsection{Hybrid Relay Selection}
\label{Outage_HRS}
For simplicity, for the analysis of HRS, we only consider the i.i.d.~fading case. Considering (\ref{eq_2.1}) and (\ref{eq_2.1.1}), the outage probability of HRS can be written as
\begin{eqnarray}\label{eq_8}
   P_{out}^{HRS}=P_{MMRS} P_{out}^{MMRS} +P_{BRS} P_{out}^{BRS},
\end{eqnarray}
where $P_{BRS}$ and $P_{MMRS}=1-P_{BRS}$ are the probabilities that BRS (i.e., $\overline{br}=\overline{bt}=b$) and  MMRS  (i.e., $\overline{br}=br$ and $\overline{bt}=bt$) are used in HRS, respectively. Since $P_{out}^{BRS}$ and $P_{out}^{MMRS}$ are already known
from the previous two subsections, we only have to compute $P_{BRS}$ (or $P_{MMRS}$) for evaluation of $P_{out}^{HRS}$.

Clearly, if all the buffers are either full or empty, BRS is used all the time and $P_{BRS}=1$ (and hence $P_{MMRS}=0$). In order to compute
$P_{BRS}$ for the more interesting case where at least one buffer is neither full nor empty,  we model the possible states of the buffers and the transitions between the states as a Markov chain. Let $S_i\triangleq X_{1}X_{2} \cdots X_{N}$ denote the $i$th state in the
Markov chain, where $X_{j}$, $j=1,\ldots,N$, represents the number of full elements in the $j$th buffer. Let $P_{BRS,i}$ denote the probability of using BRS in state $S_i$. Then, $P_{BRS}$ can be written as
\begin{eqnarray}\label{eq_9}
   P_{BRS}=\sum_{i=1}^{N_s} P_{BRS,i} P_{S_i},
\end{eqnarray}
where $P_{S_i}$ and $N_s$ denote the probability of being in state $S_i$ and the total number of states, respectively. Since the buffer size is finite and the total number of buffer elements across all relays that are full is constant,
each state has to meet the following two constraints
\begin{align}
  &\sum_{i=1}^{N} X_{i}=N_e, \label{eq_10.1}\\
  &0\leq X_{i}\leq L_b-1, ~~ i\in\{1,\ldots ,N\},  \label{eq_10.2}
\end{align}
where $N_e\triangleq \sum_{i=1}^{N}N_{e,i}$ is the total number of full elements of all buffers.

The probability of transition from one state to another state is $1/N^2$. This can be seen from the fact that for a two-hop relay network with $N$ relays there are $N^2$ possible selections of the relays for reception and transmission. Since we assume that all channels are i.i.d.,
the probability of each selection is $1/N^2$. Given that the status of the buffers changes only if the relays selected for reception and transmission are different, each transition from one state to another state corresponds to only one selection, resulting in a transition probability of
$1/N^2$. On the other hand, if any one of the $N$ available relays is selected for reception and transmission, the states of the buffers remain unchanged. Thus, there is more than one selection that allows the buffers to remain in the same state and the probability
that the buffer remains in the same state is generally larger than $1/N^2$.
\begin{proposition}\label{prp1} The state transition matrix $\mathbf{P}$ of the Markov chain that models the buffer states is a doubly stochastic matrix\footnote{A doubly stochastic matrix is a square matrix for which the sum of the elements in each of its rows and columns
is 1 \cite{Ibe_B08}.}.
\end{proposition}
\begin{proof} For any Markov chain $\sum_{j=1}^{N}p_{ij}=1$ holds, where $p_{ij}\triangleq [\mathbf{P}]_{ij}$ is the transition probability from state $S_i$ to state $S_j$. Furthermore, for the considered case, all
transition probabilities from one state to another state are equal to $1/N^2$. If there is a transition from state $S_i$ to state $S_j$, there is also a transition from state $S_j$ to state $S_i$ and the probability
of both transitions is $1/N^2$. Thus, the transition matrix $\textbf{P}$ is symmetric and $\sum_{i=1}^{N}p_{ij}=1$ holds. Hence, the transition matrix is doubly stochastic, and the proof is complete.
\end{proof}
\begin{lemma}[{\cite[Page 65]{Ibe_B08}}]\label{thm3} For a doubly stochastic transition matrix, the stationary distribution is uniform, i.e., all the states are equally likely. For an $N_s$-state Markov chain, the probability of being in state
$S_i$, $i=1,\ldots ,N_s$, is $P_{S_i}=\frac{1}{N_s}$, regardless the initial state.
\end{lemma}
From \textit{Lemma} \ref{thm3}, (\ref{eq_9}) reduces to
\begin{eqnarray}\label{eq_10.3}
   P_{BRS}=\frac{1}{N_s}\sum_{i=1}^{N_s} P_{BRS,i}.
\end{eqnarray}
Since the computation of $P_{BRS,i}$ is difficult in the general case, we first consider an example to illustrate the main idea.
\begin{example}\label{exp1} Let us consider a relay network with $N=2$ relays and the buffer at each relay is of size $L_b=4$ and half of the buffer elements are full, i.e., $N_{e}=4$. Fig. \ref{fig_exp_net} depicts the block diagram of the network.
The states of the corresponding Markov chain have to satisfy the constraints in (\ref{eq_10.1}) and (\ref{eq_10.2}), i.e.,
\begin{eqnarray}\label{eq_10.4}
  X_1+X_2=4, ~~ X_1\leq 3,~~\mathrm{and}~~ X_2\leq 3.
\end{eqnarray}
\begin{figure}[t]
    \centering
    \includegraphics[width=8cm,height=4cm]{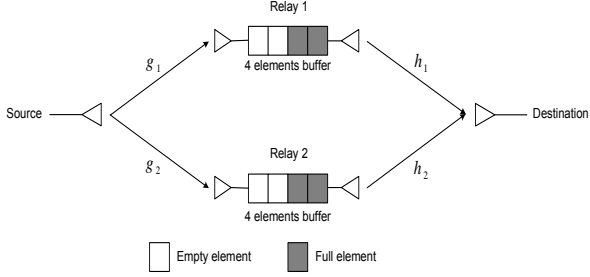}
    \caption{Relay network with one source, one destination, and two relays where each relay is equipped with a four-element buffer.}\label{fig_exp_net}
\end{figure}
\begin{figure}[t]
    \centering
    \includegraphics[width=7.5cm,height=3.5cm]{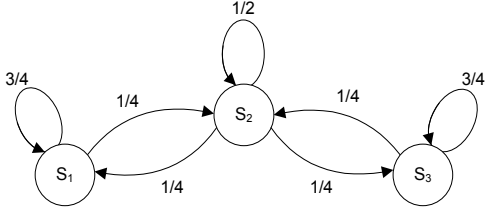}
    \caption{State diagram of the Markov chain representing the states of the buffers and the transitions between them.}\label{fig_exp_diag}
\end{figure}
Therefore, we have $N_s=3$ possible states for the Markov chain, $X_1X_2\in\{13,22,31\}$, and the probability of transition from one state to another is $\frac{1}{N^2}=\frac{1}{4}$. Let $S_1=13$, $S_2=22$, and $S_3=31$. The state diagram of this Markov chain is shown in Fig. \ref{fig_exp_diag}. The corresponding state transition matrix is given by
\begin{eqnarray}\label{eq_11}
  \mathbf{P}=\left[
    \begin{array}{ccc}
      3/4 & 1/4 & 0 \\
      1/4 & 1/2 & 1/4 \\
      0 & 1/4 & 3/4 \\
    \end{array}
  \right].
\end{eqnarray}
As expected, the transition matrix $\mathbf{P}$ is doubly stochastic. Therefore, from \textit{Lemma} \ref{thm3} the probability of each state is $\frac{1}{N_s}=\frac{1}{3}$.

Let us now define matrix
$\mathbf{D}\triangleq\left[
    \begin{array}{cc}
      a_1 & b_1  \\
      a_2 & b_2  \\
    \end{array}
  \right]$,
where the first and second column represents the relays selected for reception and transmission, respectively. If $a_i=1$, relay $R_i$ is selected for reception and if $b_i=1$, relay $R_i$ is selected for transmission, otherwise both $a_i$ and $b_i$ are zero.
Note that since only one relay is selected for reception and transmission, respectively, in each column of $\mathbf{D}$ only one element is equal to one and all other elements are zero. Hence, $\mathbf{D}$ can assume $N^2=4$ different values: $\mathbf{D}_1
\triangleq\left[ \begin{array}{cc} 1 & 1 \\ 0 & 0\\ \end{array}\right]$, $\mathbf{D}_2\triangleq\left[ \begin{array}{cc} 0 & 0 \\ 1 & 1\\ \end{array}\right]$, $\mathbf{D}_3\triangleq\left[ \begin{array}{cc} 1 & 0 \\ 0 & 1\\ \end{array}\right]$, and $\mathbf{D}_4\triangleq\left[
\begin{array}{cc} 0 & 1 \\ 1 & 0\\ \end{array}\right]$ and the probability of occurrence of each value is $\frac{1}{N^2}=\frac{1}{4}$. Now, we are ready to compute $P_{BRS,i}$ for each state.

\begin{enumerate}
  \item State $S_1=13$: For $\mathbf{D}_1$ and $\mathbf{D}_3$, based on (\ref{eq_2.1}) and (\ref{eq_2.1.1}) MMRS is used, and hence $P_{BRS,1}^{1}=P_{BRS,1}^{3}=0$. For $\mathbf{D}_2$ and $\mathbf{D}_4$, again based  on
  (\ref{eq_2.1}) and (\ref{eq_2.1.1}), BRS is used, and hence $P_{BRS,1}^{2}=P_{BRS,1}^{4}=\frac{1}{N^2}=\frac{1}{4}$. Summing up the probabilities of using BRS in state $S_1$, we obtain $P_{BRS,1}=\sum_{i=1}^{4}P_{BRS,1}^{i}=\frac{1}{2}$.
  \item State $S_2=22$: In this case, $P_{BRS,2}=0$, since the buffers are neither full nor empty.
  \item State $S_3=31$: This state is symmetric to state $S_1$. Thus, $P_{BRS,3}=\frac{1}{2}$.
\end{enumerate}

Finally, the total probability of using BRS is obtained as
\begin{eqnarray}\label{eq_13}
   P_{BRS}=\frac{1}{N_s}\sum_{i=1}^{N_s} P_{BRS,i}=\frac{1}{3}(\frac{1}{2}+\frac{1}{2})=\frac{1}{3}.
\end{eqnarray}
\end{example}
Let us now return to the general case.
\begin{proposition}\label{prp2} Let $N_{F,i}$, and  $N_{E,i}$ denote the number of full and empty buffers for state $S_i$, respectively. Note that the buffer of relay $R_i$ is considered full if $N_{e,i}=L_b-1$. Then, the probability of using BRS in state $S_i$ is given by
\begin{eqnarray}\label{eq_14}
   P_{BRS,i}&=&\frac{1}{N^2}\left((N_{F,i}+N_{E,i})N-N_{F,i}N_{E,i}\right).
\end{eqnarray}
\end{proposition}
\begin{proof}
As mentioned before, for $N$ relays, there are $N^2$ possible selections of the relays for reception and transmission and the probability of each selection is $1/N^2$. Therefore, according to (\ref{eq_2.1}) and (\ref{eq_2.1.1}), if we have $N_{F,i}$ full buffers in state $S_i$, there are $N_{F,i}N$ selections for which BRS is used. Moreover, if sate $S_i$ has $N_{E,i}$ empty buffers, from the remaining $N^2-N_{F,i}N$ possible selections, there are $(N-N_{F,i})N_{E,i}$ selections in which BRS is used. Therefore, since each selection has a probability of occurrence of $1/N^2$, the probability of using BRS in state $S_i$ is given by: $ P_{BRS,i}=\frac{1}{N^2}\left(N_{F,i}N+(N-N_{F,i})N_{E,i}\right)=\frac{1}{N^2}\left((N_{F,i}+N_{E,i})N-N_{F,i}N_{E,i}\right).$ This concludes the proof.
\end{proof}

\textit{Computation of $N_s$, $N_{F,i}$, and $N_{E,i}$}: For computation of $P_{BRS}$, the number of possible buffer states, $N_s$, has to be determined. It does not seem possible to obtain a general formula for $N_s$ valid for
any $N$, $L_b$, and $N_{e}$. However, for $N=2$ and $N=3$, the number of states can be calculated in closed form. In particular, for $N=2$, we obtain
\begin{equation}\label{eq_15}
    N_s=\left\{
        \begin{array}{l}
            N_e+1,~~ \mathrm{if}~N_e\leq L_b-1  \\
            2L_b-N_e-1,~~\mathrm{otherwise} \\
        \end{array}
        \right.
\end{equation}
and, for $N=3$, we have
\begin{equation}\label{eq_16}
    N_s=\sum_{i=(N_e-L_b)^++1}^{N_e+1}\Big(i-2(i-L_b)^+\Big)^+,
\end{equation}
where $(x)^+\triangleq \max\{x,0\}$.

For the general case, for a given total number of full buffer elements, $N_e$, and a given size of the buffers, $L_b$, the number of states can be obtained algorithmically as the number of all possible combination of $X_1X_2...X_N$ that satisfy (\ref{eq_10.1}) and (\ref{eq_10.2}).

Given the states of the Markov chain, it is easy to obtain the number of empty and full buffers, $N_{E,i}$ and $N_{F,i}$, for each state $S_i$, $i=1,\ldots,N_s$.
Thus, $P_{BRS}$ can be computed based on (\ref{eq_10.3}) and (\ref{eq_14}).

\textit{Computation of $P_{out}^{HRS}$}: Given $P_{out}^{BRS}$ (\ref{eq_2.3}), $P_{out}^{MMRS}$ (\ref{eq_3}), and $P_{BRS}$, $P_{out}^{HRS}$ can be computed using (\ref{eq_8}).

In the asymptotic case of high SNR, the outage probability of HRS can be expressed as
\begin{eqnarray}\label{eq_8.1}
   P_{out}^{HRS}\approx\left(\left(2P_{MMRS}+2^N P_{BRS}\right)^\frac{1}{N} \frac{\gamma}{\overline{\gamma}}\right)^{N}.
\end{eqnarray}
Therefore, the diversity gain of HRS is $G_d^{HRS}=N$ and the coding gain is $G_c^{HRS}=\left(2P_{MMRS} +2^N P_{BRS}\right)^{-\frac{1}{N}}$, i.e., the coding gain of HRS increases with increasing $N$ and increasing $P_{MMRS}$.
%
\section{Numerical results}
\label{results}
In this section, we assess the performance of the proposed MMRS and HRS schemes and compare it with that of BRS \cite{Bletsas_KRL_JSAC06}. Throughout this section, a target rate $R=1 \mathrm{bit/sec/Hz}$ is assumed for outage probability calculation. Furthermore, we consider the i.i.d.~case where all the channel coefficients are modeled as zero-mean complex Gaussian random  variables with variance $\sigma_{g_i}^2=\sigma_{h_i}^2=1$, and, hence $\overline{\gamma}_{g_i}=\overline{\gamma}_{h_i}=\overline{\gamma}= \frac{E_s}{N_0}=\mathrm{SNR}$.
\begin{figure}[t]
    \centering
    \includegraphics[width=9.5cm,height=7.5cm]{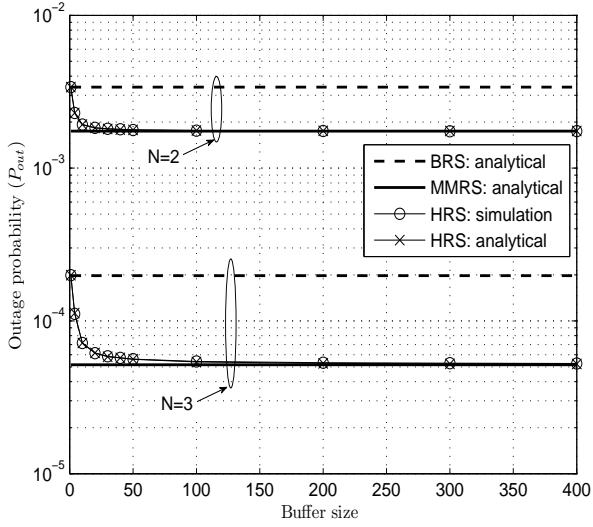}
    \vspace*{-7mm}
    \caption{Outage probability vs.~buffer size, $L_b$, for $N=2$ and $N=3$ relays. Half of the buffer elements are full, i.e., $N_e=\lceil NL_b/2\rceil$, SNR = 20 dB, and the target rate is $R=1 \mathrm{bit/sec/Hz}$. The analytical results were obtained from (\ref{eq_2.3}), (\ref{eq_3}), and (\ref{eq_8}).}\label{fig_Pout_Buf_size}
\end{figure}

Fig.~\ref{fig_Pout_Buf_size} shows the outage probability vs.~the buffers size, $L_b$, for $N=2$ and $N=3$ relays and SNR = 20 dB. We assume that half of the buffer elements are full, i.e., $N_e=\lceil NL_b/2\rceil$, where $\lceil x\rceil$ denotes the smallest
integer larger than or equal to $x$. Fig. \ref{fig_Pout_Buf_size} shows that for buffers of size $L_b=1$ HRS behaves like BRS. As the buffer size increases, the performance of HRS converges to that of MMRS. In fact, for $L_b\geq10$ and $L_b\geq30$ HRS achieves
a performance very similar to that of MMRS for $N=2$ and $N=3$, respectively. We also note that the analytical results obtained based on the derivation in Section \ref{Outage_analysis} are in perfect agreement with the simulation results.

In Fig.~\ref{fig_Pout_Buf_Numb_full_elem}, we investigate the impact of the average number of full buffer elements per relay on the outage probability of HRS. We consider buffers with $L_b=100$ elements and assume again SNR = 20 dB. As can be observed,
the best performance is achieved if the buffers are half full on average, i.e., the total number of full buffer elements is $N_e=\lceil NL_b/2\rceil$. In this case, the probability of having empty or full buffers is minimized and the probability that MMRS is used is maximized.

Fig.~\ref{fig_Pout_SNR_N} depicts the outage probability of MMRS, HRS, and BRS vs.~the average SNR for various numbers of relays. For HRS, the buffer size is $L_b=30$ and half of the buffer elements are full. As expected, all considered relay selection schemes
achieve a diversity gain of $G_d=N$. However, the coding gain advantage of MMRS and HRS increases with increasing number of relays. For $N=1$, MMRS and HRS are identical to BRS (in fact, no selection takes place in this case). For $N=2$, 3, and 5, the
asymptotic SNR gain of MMRS compared to BRS is 1.5 dB, 2.0 dB, and 2.4 dB, respectively. The gap between MMRS and HRS increases slightly with increasing $N$ indicating that the buffer size has to increase with increasing $N$ to keep the gap between
MMRS and HRS constant.

\begin{figure}[t]
    \centering
    \includegraphics[width=9.5cm,height=7.5cm]{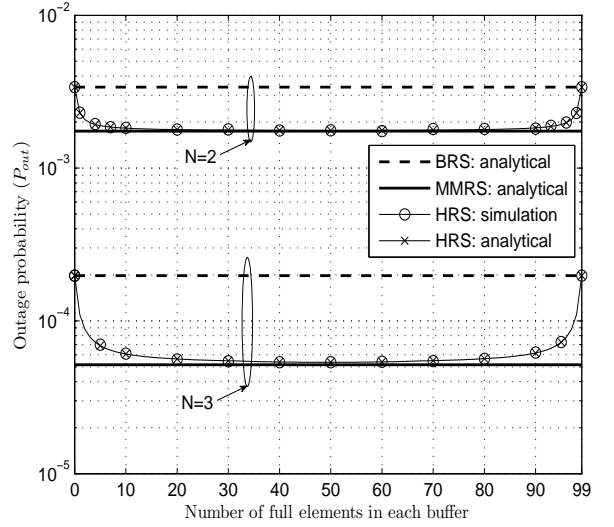}
      \vspace*{-7mm}
    \caption{Outage probability vs.~average number of full buffer elements per relay for $N=2$ and $N=3$.  $L_b=100$, SNR = 20 dB, and the target rate is $R=1 \mathrm{bit/sec/Hz}$. The analytical results were obtained from
    (\ref{eq_2.3}), (\ref{eq_3}), and (\ref{eq_8}).} \label{fig_Pout_Buf_Numb_full_elem}
\end{figure}

In Fig.~\ref{fig_Delay_Buf_Size}, we investigate the effect of the relay buffer size, $L_b$, on the average delay introduced by HRS. Half of the buffer elements are full, i.e., $N_e=\lceil NL_b/2\rceil$, and SNR = 15 dB. The delay is zero if the packet is not stored at the relay
and reaches the destination in two consecutive hops as in BRS. The delay in Fig.~\ref{fig_Delay_Buf_Size} corresponds to the number of transmission interval durations (corresponding to two hops and consequently two packet durations) that a packet is stored at a relay
before it arrives at the destination, i.e., a delay of one means that the packet is stored at the relay in one transmission interval and retransmitted in the next. Fig.~\ref{fig_Delay_Buf_Size} shows that, as expected, the average delay increases with increasing buffer size.
In fact, for sufficiently large buffer sizes the average delay can be approximated as $NL_b/2$. Thus, considering the results in Figs.~\ref{fig_Pout_Buf_size} and \ref{fig_Pout_SNR_N}, delays of less than 50 (100) transmission intervals are sufficient for HRS to closely
approach the performance of MMRS for $N=3$ ($N=5$) relays.
\begin{figure}[t]
    \centering
    \includegraphics[width=9.5cm,height=7.5cm]{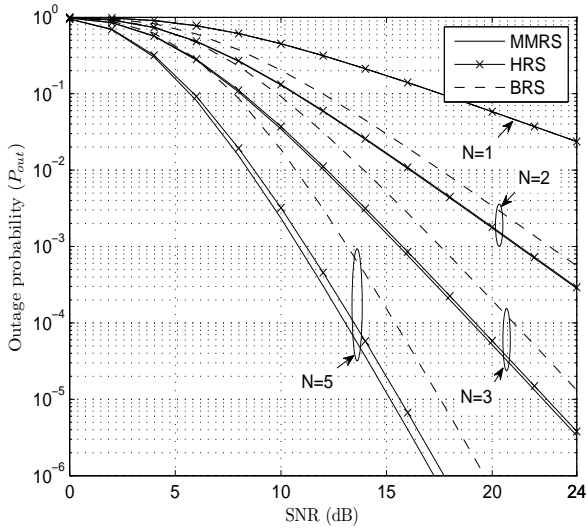}
      \vspace*{-7mm}
    \caption{Outage probability of MMRS, HRS, and BRS vs.~average SNR for different numbers of relays. For HRS, the buffer size is $L_b=30$ and half of the buffer elements are full, i.e., $N_e=\lceil NL_b/2\rceil$. The target rate is $R=1 \mathrm{bit/sec/Hz}$. Analytical results
    obtained from (\ref{eq_2.3}), (\ref{eq_3}), and (\ref{eq_8}) are shown.}\label{fig_Pout_SNR_N}
\end{figure}
\section{Conclusion}
\label{conclusion}
In this paper, we proposed two new relay selection schemes for relays with buffers. The first scheme, MMRS, always selects the relays with the best source-relay and the best relay-destination channels for reception and transmission, respectively, and operates under
the assumption that the buffers at the relays are neither full nor empty. Since this assumption is not practical for finite buffers, we proposed a second scheme, HRS, which employs MMRS if the buffer of the relay selected for reception is not full and the buffer of the
relay selected for transmission is not empty, and conventional BRS otherwise. We have analyzed the outage probability of MMRS and HRS and established that while they have the same diversity gain as BRS, they achieve a coding gain advantage of up to 3 dB.
More importantly, we showed that, for $N=3$ relays, HRS can achieve a coding gain advantage of 2 dB compared to BRS if an average delay of 50 transmission intervals can be afforded.

Finally, we note that relays with buffers add additional flexibility to cooperative diversity systems. While, in this paper, we used this flexibility to improve the performance of relay selection, exploring other scenarios (e.g.~interference
avoidance) where relays with buffers may be advantageous in cooperative networks is an interesting topic for future work.
\begin{figure}[tbp]
    \centering
    \includegraphics[width=9.5cm,height=7.5cm]{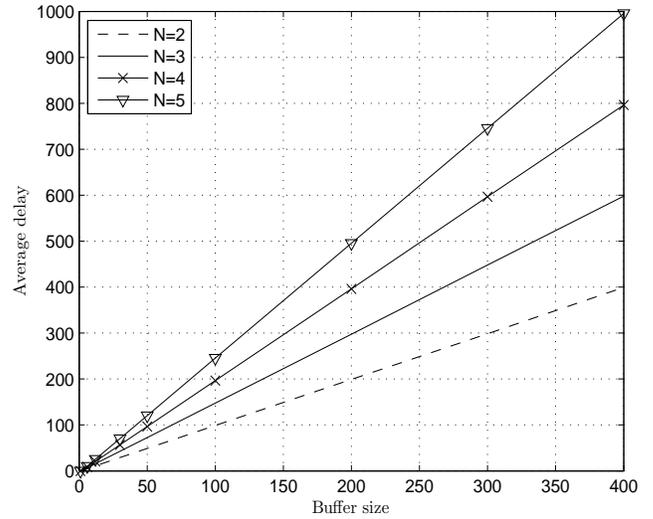}
      \vspace*{-7mm}
    \caption{Average delay vs.~buffer size $L_b$ for HRS. Half of the buffer elements are full, i.e., $N_e=\lceil NL_b/2\rceil$, SNR = 15 dB, and the target rate is $R=1 \mathrm{bit/sec/Hz}$. Simulation results are shown.}\label{fig_Delay_Buf_Size}
\end{figure}
\bibliographystyle{IEEEtran}
\bibliography{refs}
\end{document}